\documentclass[a4paper,10pt]{article}

\usepackage{mathtools,xspace,latexsym}
\usepackage{amsmath, amssymb, amsthm}
\usepackage{hyperref, cleveref, authblk}
\newcommand{\RB}{\ensuremath{\mathbb R}\xspace}

\newcommand{\BB}{\ensuremath{\mathcal B}\xspace}
\newcommand{\CC}{\ensuremath{\mathcal C}\xspace}

\newcommand{\EE}{\ensuremath{\mathcal E}\xspace}

\newcommand{\GG}{\ensuremath{\mathcal G}\xspace}
\newcommand{\HH}{\ensuremath{\mathcal H}\xspace}

\newcommand{\KK}{\ensuremath{\mathcal K}\xspace}

\newcommand{\OO}{\ensuremath{\mathcal O}\xspace}

\newcommand{\UU}{\ensuremath{\mathcal U}\xspace}
\newcommand{\VV}{\ensuremath{\mathcal V}\xspace}
\newcommand{\WW}{\ensuremath{\mathcal W}\xspace}

\newcommand{\ttt}{\ensuremath{\mathfrak t}\xspace}

\newcommand{\www}{\ensuremath{\mathfrak w}\xspace}

\newcommand{\calP}{\mathcal{P}}
\newcommand{\calY}{\mathcal{Y}}
\newcommand{\calp}{\mathcal{P}}

\newcommand{\tbl}{\mathsf{DP}}

\newcommand{\CN}{{\sc Constructive Control over Network}\xspace}

\newcommand{\DCN}{{\sc Destructive Control over Network}\xspace}

\newcommand{\BCN}{{\sc Budgetless Constructive Control over Network}\xspace}

\newcommand{\BDCN}{{\sc Budgetless Destructive Control over Network}\xspace}

\newcommand{\CB}{{\sc CanBeat}}

\newcommand{\XThC}{{\sc Regular Exact 3-cover}\xspace}

\newtheorem{theorem}{Theorem}
\newtheorem{definition}{Definition}
\newtheorem{probdefinition}{\bf Problem Definition}
\newtheorem{claim}{Claim}
\providecommand{\keywords}[1]
{
  \small	
  \textbf{\textit{Keywords---}} #1
}

\title{Voter Participation Control in Online Polls}
\author[1]{Koustav De \\ email: \href{mailto:koustavde7@kgpian.iitkgp.ac.in}{koustavde7@kgpian.iitkgp.ac.in}}
\author[1]{\\Palash Dey \\ email: \href{mailto:palash.dey@cse.iitkgp.ac.in}{palash.dey@cse.iitkgp.ac.in}}
\author[1]{\\Swagato Sanyal \\ email: \href{mailto:swagato@cse.iitkgp.ac.in}{swagato@cse.iitkgp.ac.in}}
\affil[1]{\small Department of Computer Science and Engineering, Indian Institute of Technology Kharagpur}

\begin{document}
\maketitle
\begin{abstract}
News outlets, surveyors, and other organizations often conduct polls on social networks to gain insights into public opinion. Such a poll is typically started by someone on a social network who sends it to her friends. If a person participates in the poll, the poll information gets published on her wall, which in turn enables her friends to participate, and the process continues. Eventually, a subset of the population participates in the poll, and the pollster learns the outcome of that poll. We initiate the study of a new but natural type of election control in such online elections.

We study how difficult/easy it is to sway the outcome of such polls in one's favor/against (aka constructive vs destructive) by any malicious influencer who nudges/bribes people for seemingly harmless actions like non-participation. These questions are important from the standpoint of studying the power of resistance of online voting against malicious behavior. The destructive version is also important to quantify the robustness of the winner of an online voting. We show that both problems are computationally intractable even if the election is over only two candidates and the influencer has an infinite amount of money to spend (that is, every voter can be persuaded to not participate). We strengthen this result by proving that the computational task remains substantially challenging even if the underlying network is a tree. Finally, we show that there is a polynomial-time algorithm for the constructive version of the problem when we have $\OO(1)$ candidates, and the treewidth of the underlying graph is $\OO(1)$; the algorithm for the destructive version does not even need to assume $\OO(1)$ number of candidates. Hence, we observe that the destructive version is computationally easier than the constructive version.
\end{abstract}





         
\newcommand{\BibTeX}{\rm B\kern-.05em{\sc i\kern-.025em b}\kern-.08em\TeX}








\section{Introduction}
\label{sec:introduction}
Voting is arguably the most widely used tool when a set of people needs to decide on one alternative/candidate. It also has widespread applications in AI also~\cite{DBLP:conf/aaai/PennockHG00,DBLP:journals/tcbb/JacksonSA08}. The study of various election malpractices and their computational complexity has been one of the core research focus in computational social choice \cite{DBLP:reference/choice/2016,BARTHOLDI199227}. Election Control is one of the most well-studied malpractice where someone tries to change the structure of the election, for example, the set of voters, the set of candidates, etc. Indeed, in many voting scenarios, the key personalities of an institution often have the power to decide the committee members, i.e., the voters, who will eventually make some crucial decisions. Another type of election control that is often observed in political elections is adding some ``spoiler'' candidates to create confusion among voters or persuade some candidates to withdraw who are poorly performing in pre-election polls.

This work focuses on online voting. In this setting, there is a social network on the set of voters. A voter initiates an election (online poll or survey, for example) and asks her friends to participate in it. Depending on the social network platform, there could be various ways to do this job. For example, on a Facebook network, a voter can post a poll on her wall that, when her friends see it, participate in it. When their friends participate in that poll, their friends will see it and can participate in it. If everyone who sees the poll participates in it, then, if the social network is connected, then everyone in the network participates in it. However, this is rarely the case because some people may not participate in the poll even after seeing it. Hence, only a connected subset of people, including the poll initiator, participates in that poll. Since the outcome of such polls can be used for important decision-making tasks, it may be possible that such elections come under attack. We study an important type of attack on such elections: persuading some voters not to participate in the election, thereby controlling the set of voters that eventually participate in the election. Such a control attack can be done for various purposes, for example, confusing the pollster, helping some candidate to win/lose the election, etc. In this paper, we focus on the objective of helping some candidate win/lose the election under a complexity-theoretic lens.

\section{Related Work}
\label{sec:related}

In this section, we discuss some of the most important work in electoral control. Bartholdi et al. \cite{BARTHOLDI199227} initiated the study of electoral control and looked into it from various computational perspectives. The chair may exercise control over the candidate set by removing up to $k$ candidates from the election or by inserting new candidates from a list of spoiler candidates. Hemaspaandra et al. \cite{Hemaspaandra2009-HEMHEB} later defined a variant of the problem where the number of spoiler candidates that might be added by the election controller has a bound $k$. Many voting rules, for example, Fallback and Bucklin \cite{DBLP:conf/cats/ErdelyiR10}, $\text{Copeland}^{\alpha}$ \cite{10.5555/1641503.1641510}, Normalized Range Voting \cite{10.1007/s00224-012-9441-0}, SP-AV voting \cite{https://doi.org/10.1002/malq.200810020} and Schulze Voting \cite{DBLP:conf/ijcai/MentonS13} are resistant to all types of constructive control. Bodlaender \cite{10.5555/646242.681417} showed that intractable computational problems on graphs usually become tractable if the treewidth is bounded by a constant. The idea of safety in electoral control issues has been studied previously by Slinko and White, who examined the class of social choice functions that can be safely manipulable and put forth the idea of safety in the context of manipulation \cite{slinko2008nondictatorial,DBLP:journals/scw/SlinkoW14}. The authors Hazon and Elkind \cite{10.1007/978-3-642-16170-4_19} and Ianovski et al. \cite{DBLP:conf/ijcai/IanovskiYEW11} have looked into the complexity of safely manipulating popular voting rules. Faliszewski et al. \cite{DBLP:conf/aaai/FaliszewskiHH06} have shown that in a bribery problem, a briber, who can be an election controller, can change the minimum number of preferences to make way for a preferred candidate to win the election. Bredereck and Elkind \cite{ijcai2017p124} analysed the computational complexity of bribery and control by adding/deleting links between users on an online social network and altering the order in which voters update their opinions. Goles and Olivos \cite{DBLP:journals/dm/GolesO80} showed that a sequence of at most $\OO(n^2)$ synchronous updates, where $n$ is the number of voters, always converges to a stable state. Frischknecht et al. \cite{10.1007/978-3-642-41527-2_30} strengthened the tightness of the stated result. Wilder and Vorobeychik \cite{10.5555/3237383.3237428} showed hardness, inapproximability, and algorithmic results for the problem of constructive control and destructive control through social influence. Opinion dynamics and social choice have been extensively studied by \cite{6736702,10.1007/s10458-013-9230-4,10.5555/2832415.2832532,10.5555/2936924.2936979}. We refer to \cite{grandi2017social} for a comprehensive but not-so-recent survey of opinion diffusion.

\section{Preliminaries and problem definitions}
\label{sec:prelims}
We consider the plurality voting rule in this paper. For an integer $g$, we denote the set $\{ 1,\ldots, g \}$ by $\left[ g \right]$ and $\left[ g \right]_{0} = \left[ g \right] \cup \{ 0 \}$. Let $\CC = \{c_1,\ldots,c_m\}$ be a set of candidates and $\VV = \{v_1 , \ldots, v_n\}$ be a set of voters. Each voter $v_i$ votes for a candidate in $\CC$ that is given by a voting function $\tau:\VV\to\CC$. Thus, the voter $v\in \VV$ votes for the candidate $\tau(v) \in \CC$. The triple $(\CC, \VV, \tau)$ is called an election. A candidate $c \in \CC$ is declared a winner of the election if $c\in\arg \max_{c' \in \CC}|\{v \in \VV \mid \tau(v)=c'\}|$. If there is only one winner $c$ of an election, then $c$ is said to win the election uniquely. The mathematical formalism to express an election for a general voting rule is more elaborate; we do not include that in this paper.
\par We consider two problems where the election controller wants to control the winner of a given election, by deleting some votes within her budget limit. In the \CN problem, the election controller wishes to make her preferred candidate unambiguously win the election. In the \DCN problem, the election controller wishes to stop a specific candidate to win the election. The problems are formally defined in \Cref{def:CN}.
\begin{probdefinition}[\sc{Constructive and Destructive Control over Network}]
\label{def:CN}
  We are given a set $\CC=\{ c_1,\ldots, c_m \}$ of $m$ candidates, a set $\VV=\{ v_1, \ldots, v_n \}$ of $n$ voters, a voting function $\tau: \VV \to \CC$, an undirected graph $\GG=(\VV,\EE)$ whose vertices are the voters, a target candidate $c \in \CC$ of the controller, a voter $x\in\VV$ who conducts the election, a cost function $\pi:\VV \longrightarrow \RB_{\ge 0}$, and a budget \BB of the controller. We extend the definition of $\pi$ to subsets of $\VV$, and define the cost $\pi(\KK)$ of a subset $\KK \subseteq \VV$ as $\sum_{v \in \KK}\pi(v)$. We say that a subset $\WW \subseteq \VV \setminus \{ x \}$ is ``budget feasible'' if $\sum_{v \in \WW} \pi(v) \le \BB$. For a budget feasible set $\WW$, let $\HH_\WW$ be the set of nodes reachable from $x$ in $\GG \setminus \WW$.
  
  The \CN problem asks whether there exists a budget feasible set $\WW$ such that $c$ wins uniquely in the restricted election $(\CC, \VV, \tau)$ where only votes of $\HH_\WW$ are counted.

  The \DCN problem asks whether there exists a budget feasible set $\WW$ such that a candidate other than $c$ unambiguously wins in the restricted election $(\CC, \VV, \tau)$ where only votes of $\HH_\WW$ are counted.
\end{probdefinition}
A generic input to both \CN and \DCN is a tuple $(\CC, \VV, \tau, \GG, c, x, \pi, \BB)$.

We also consider the special setting where the budget is infinite. In this setting, the cost function $\pi$ is irrelevant, and any subset $\WW \subseteq \VV \setminus \{ x \}$ trivially satisfies the budget constraint. We refer to the corresponding versions as \BCN and \BDCN. An instance of the budgetless versions is a tuple $(\CC, \VV, \tau, \GG, c, x)$, where the members of the tuple are as per \Cref{def:CN}.

Observe that an efficient algorithm $\mathcal{A}$ for \CN can be used to design an efficient algorithm for \DCN as follows: examine each candidate $c'$ other than $c$ in turn, and decide whether $c'$ can be made a unique winner within the given budget by invoking $\mathcal{A}$.

In some of our results, the problem \XThC is used, that we define below formally.
 \begin{definition}[\XThC]
 \label{dxef:xthc}
 For a positive integer $\ell$, let $\UU\coloneqq \{1, \ldots, 3\ell\}$. We are given $m$ subsets $S_1, \ldots, S_m$ of $\mathcal{U}$, each of cardinality $3$ such that $\cup_{i \in [m]}S_i=\UU$. Furthermore, each element in $\UU$ belongs to exactly two sets in the collection $\{S_1, \ldots, S_m\}$. Decide whether there exists a subset $A \subseteq [m]$ of size $\ell$ such that $\cup_{i \in A} S_i=\mathcal{U}$. Note that if such a set $A$ exists, then the sets $\{S_i \mid i \in A\}$ are pairwise disjoint. We denote an instance of \XThC as $(\ell, S_1, \ldots, S_m)$.
 \end{definition}
 \XThC is known to be NP-complete \cite{LT06}.


\subsection{Tree decomposition and treewidth}
\label{sec:tw}
In this section, we formally define tree decomposition and treewidth. See \cite{DBLP:books/sp/CyganFKLMPPS15} for more details.
\begin{definition}[Treewidth]
\label{def:treewidth}
        Let $\GG = (V_{\GG}, E_{\GG})$ be a graph.  A {\em tree-decomposition} of $\GG$ is a pair $(T = (V_T, E_T),\mathcal{X}=\{X_{t}\}_{t\in V_T})$,  where $T$ is a tree where every node $t\in V_T$ is assigned a subset $X_t\subseteq V_{\GG}$, called a bag,  such that the following conditions hold.
        \begin{itemize}
            \item $\bigcup_{t\in V_T}{X_t}=V_{\GG}$,
            \item for every edge $\{x,y\}\in E_{\GG}$ there is a $t\in V_T$ such that  $x,y\in X_{t}$, and
            \item for any $v\in V_{\GG}$ the subgraph of $T$ induced by the set  $\{t\mid v\in X_{t}\}$ is connected.
        \end{itemize}

        The {\em width} of a tree decomposition is $\max_{t\in V_T} |X_t| -1$. The {\em treewidth} of $\GG$, denoted by $\ttt\www(\GG)$, is the  minimum width over all tree decompositions of ${\GG}$.

        A tree decomposition  $(T,\mathcal{X})$ is called a {\em nice edge tree decomposition} (we will also refer to it as \emph{nice tree decomposition}) if $T$ is a tree rooted at some node $r$ where $ X_{r}=\emptyset$, each node of $T$ has at most two children, and each node is of one of the following kinds:
        \begin{itemize}
            \item {\bf Introduce node}: a node $t$ that has only one child $t'$ where $X_{t}\supset X_{t'}$ and  $|X_{t}|=|X_{t'}|+1$.
            \item {\bf Introduce edge node}: a node $t$ labeled with an edge $\{u, v\} \in E_{\GG}$, with only one child $t'$ such that $\{u,v\}\subseteq X_{t'}=X_{t}$. This bag is said to introduce the edge $\{u, v\}$.
            \item {\bf  Forget vertex node}: a node $t$ that has only one child $t'$  where $X_{t}\subset X_{t'}$ and  $|X_{t}|=|X_{t'}|-1$.
            \item {\bf Join node}: a node  $t$ with two children $t_{1}$ and $t_{2}$ such that $X_{t}=X_{t_{1}}=X_{t_{2}}$.
            \item {\bf Leaf node}: a node $t$ that is a leaf of $T$, and $X_{t}=\emptyset$.
        \end{itemize}
        We additionally require that every edge is introduced exactly once. One can show that a tree decomposition of width $w$ can be transformed into a nice tree decomposition of the same width $w$ and with $\OO(w |V_\GG|)$ nodes in polynomial time, see~e.g.~\cite{BODLAENDER201842,DBLP:books/sp/CyganFKLMPPS15}.


        \end{definition}
Suppose $T'$ is a nice tree decomposition of $\GG$ of width $w$, and $x \in V_{\GG}$ be a vertex of $\GG$. One can easily convert $T'$ to another tree decomposition $T$ of $\GG$ of width at most $w+1$ with the following properties:
\begin{itemize}
\item The parent of each leaf of $T$ is an Introduce node where the vertex $x$ is introduced.
\item The vertex $x$ is never forgotten. In particular, $x$ is part of the bag of each internal node of $T$.
\end{itemize}
$T$ can be obtained from $T'$ as follows:
\begin{enumerate}
\item Include $x$ in the bag of each internal node of $T'$.
\item For each leaf $\ell$, if its parent $p$ is not an Introduce node where $x$ is introduced, then create an Introduce node $t$ where $x$ is introduced, make $t$ the parent of $\ell$ and $p$ the parent of $t$.
\item If $t$ is a vertex that is not the parent of a leaf, which is either an Introduce node in $T'$ where $x$ is introduced or a Forget vertex node where $x$ is forgotten, then merge it with its child.
\end{enumerate}
Assuming that there is no join node with an empty bag in $T'$, it can be easily verified that each node of $T$ is of one of the four kinds listed in the definition of a nice tree decomposition. The only condition that $T$ violates in the definition of a nice tree decomposition is that the bag of the root is empty. Nevertheless, we abuse notation and also refer to $T$ as a nice tree decomposition of $\GG$.
\section{Our Contribution}
\label{sec:contrib}
In this work, we study the computational complexity of \CN and \DCN.

Our first result shows that \DCN admits a polynomial time algorithm for the special case where the treewidth (see \Cref{sec:tw} for related definitions) of the graph $\GG$ is a constant.
\begin{theorem}
\label{thm:easy-low-tw-dcn}
There exists an algorithm that, given an input $(\CC, \VV, \tau, \GG, c, x, \pi, \BB)$ of \DCN and a tree decomposition $T$ of the graph $\GG$ of width $w$, solves \DCN in time $w^{\OO(w)}\cdot n^{\OO(w)}\cdot \mathsf{poly}(n, m)$. In particular, \DCN admits a polynomial time algorithm when $\GG$ is a tree.
\end{theorem}
How about \CN? As discussed in \Cref{sec:prelims}, \CN is computationally at least as hard as \DCN (up to a factor of $m$). Our next result obtains a polynomial time algorithm for \CN with an assumption of the treewidth of $\GG$ being a constant as in \Cref{thm:easy-low-tw-dcn}, and an additional assumption that the number of candidates $m$ is a constant.
\begin{theorem}
\label{thm:easy-low-tw-const-cand}
There exists an algorithm that, given an input $(\CC, \VV, \tau, \GG, c, x, \pi, \BB)$ of \CN and a tree decomposition $T$ of the graph $\GG$ of width $w$, solves \CN in time $w^{\OO(w)}\cdot n^{\OO(mw)}\cdot \mathsf{poly}(n, m)$. In particular, \CN admits a polynomial time algorithm when $\GG$ is a tree, and the number of candidates $m$ is a constant.
\end{theorem}
\Cref{thm:easy-low-tw-dcn} and \Cref{thm:easy-low-tw-const-cand} assume that a tree decomposition of low treewidth is given as a part of the input. However, this assumption is not restrictive; it is known that given a graph $\GG$ with $n$ vertices it is possible to construct a tree decomposition of $\GG$ of width $\ttt\www(\GG)$ in time $\OO(f(\ttt\www(\GG))\cdot n)$, where $f(\cdot)$ is a quasi-polynomially growing function \cite{B96, DBLP:books/sp/CyganFKLMPPS15}.

Are the assumptions in \Cref{thm:easy-low-tw-dcn} and \Cref{thm:easy-low-tw-const-cand} necessary? We answer this question in the affirmative by two hardness results. The first one shows that \CN is NP-complete even when the graph $\GG$ is a tree (i.e. has treewidth $1$) and the setting is budgetless.
\begin{theorem}
\label{thm:hardness-tree}
\BCN is NP-complete for the special case where the graph $\GG$ is a tree.
\end{theorem}

Our next result shows that both the problems are NP-complete even in the special case where there are two candidates and the setting is budgetless.
\begin{theorem}
\label{thm:hardness-2cand}
\BCN and \BDCN are NP-complete even for the special case where there are two candidates (i.e.~$m=2$).
\end{theorem}
Observe that for $m=2$ the two problems are equivalent in the following sense. Let $\CC=\{0,1\}$. Then for each $b \in \CC$, $(\CC, \VV, \tau, \GG, b, x, \pi, \BB)$ is a \textsf{YES} instance of \CN \emph{if and only if} $(\CC, \VV, \tau, \GG, 1-b, x, \pi, \BB)$ is a \textsf{YES} instance of \DCN. Thus, in order to prove \Cref{thm:hardness-2cand} it suffices to prove the NP-hardness of \CN for $m=2$ in the budgetless setting.

Our hardness results presented in \Cref{thm:hardness-tree} and \Cref{thm:hardness-2cand} thus complement our algorithmic results presented in \Cref{thm:easy-low-tw-dcn} and \Cref{thm:easy-low-tw-const-cand}, and establish the unavoidability of the assumptions made in them. This paper presents a comprehensive study of the computational complexity of \CN and \DCN. In \Cref{sec:conclusion} we state some open questions for future research.
\section{Algorithm for \CN}
\label{algo:cn}
In this section, we prove \Cref{thm:easy-low-tw-const-cand}. Our proof exhibits an explicit algorithm for \CN based on dynamic programming. We assume that the tree decomposition $T$ of $\GG$ supplied as input is \emph{nice} (see \Cref{sec:tw}). Our DP table has a set of cells for each node $t$ of $T$. Each cell is additionally  indexed by a partition $\calp=(\calp_1, \ldots, \calp_\ell)$ of a subset of the bag $X_t$ of $t$, and a tuple of numbers $\overline{n}=(n_{i,j})_{i \in [\ell], j \in [m]}$. The vertices and edges of $\GG$, introduced in the subtree of $T$ rooted at $t$, induce a subgraph $G_t$ of $\GG$. In each step, the algorithm computes a cheapest set of vertices of $G_t$. Removal of such a cheapest set of vertices from $G_t$ breaks up $G_t$ into connected components whose intersections with the bag $X_t$ are given by $\calp$, and whose vote counts for each candidate are given by $\overline{n}$. Once the table is filled up, the table entries for the root of $T$ can be consulted to compute the output of the algorithm. We now proceed to the formal proof.
\begin{proof}[Proof of \Cref{thm:easy-low-tw-const-cand}]
Let $(\CC,\VV,\tau,\GG,c, x, \pi,\BB)$ be an instance of \CN. Let $x$ vote for the candidate $c'$, i.e., $\tau(x)=c'$. Let $T$ be a nice tree decomposition of $\GG$ of width $w$. Let $r$ be the root of $T$. We assume that the parent of each leaf is an \emph{Introduce node} where the vertex $x$ is introduced and that $x$ is never forgotten. In particular, $x$ is in the bag of the root of $T$. See the related discussion in \Cref{sec:prelims}.
Consider a vertex $t \in V_T$. Let $V_t$ be the union of the bags of all vertices in the subtree of $T$ rooted at $t$. Let $E_t$ be the set of all edges introduced in the subtree of $T$ rooted at $t$. Define the graph $G_t\coloneqq (V_t, E_t)$.

Consider deleting a subset $\WW \subseteq V_t \setminus \{x\}$ of vertices from $G_t$. The resulting graph, which we will denote by $G_t \setminus \WW$, will have its vertex set partitioned into connected components. Let $\ell$ of those components, denoted by $\calY^{(\WW,t)}_1,\ldots,\calY^{(\WW,t)}_\ell$, have non-empty intersections with the bag $X_t$ of $t$. For $i=1,\ldots,\ell$ define $\calP^{(\WW,t)}_i\coloneqq \calY^{(\WW,t)}_i\cap X_t$. Let $\calP^{(\WW,t)}\coloneqq (\calP^{(\WW,t)}_1,\ldots,\calP^{(\WW,t)}_\ell)$ be the partition of $X_t\setminus \WW$ defined by the $\calp^{(\WW,t)}_i$s. By definition, each $\calP^{(\WW,t)}_i$ is non-empty, which implies that $\ell \leq |X_t| \leq w+1$.

We now describe the entries of our dynamic programming table, $\tbl$. Each entry is indexed by a tuple $(t, \calp, \overline{n})$, whose components are described below.
\begin{itemize}
\item $t \in V_T$ is a vertex of $T$.
\item $\calP=(\calP_1,\ldots,\calp_\ell)$ is a partition of a subset $S$ of $X_t$ into non-empty parts $\calp_i$. If $S=\emptyset$, then $\calP$ is the empty partition that we denote by $\bot$.
\item If $S\neq \emptyset$, then $\overline{n}$ is a tuple $(n_{i,j})_{i\in [\ell], j \in [m]}$ of numbers, where each $n_{i,j}$ is a non-negative integer in $\left[ n \right]_{0}$. If $S=\emptyset$, define $\overline{n}$ to be the empty sequence that we denote by $\epsilon$.
\end{itemize}
After our algorithm terminates, $\tbl[t, \calp, \overline{n}]$ will contain the minimum cost of any subset $\WW \subseteq V_t \setminus \{x\}$ which satisfies the following conditions:
\begin{enumerate}
\item $\calP^{(\WW, t)}=\calp$, and
\item For each $i \in [\ell]$ and $j \in [m]$, the number of votes garnered by the candidate $j$ in $\calY^{(\WW, t)}_i$ is $n_{i,j}$.
\end{enumerate}
We say that a $\WW$ that satisfies the above two conditions induces the pair $(\calP, \overline{n})$ on $(X_t, G_t)$. If there is no such $\WW$, $\tbl[t, \calp, \overline{n}]$ will be $\infty$.

The number of entries in our table is at most $w^{\OO(w)}\cdot n^{\OO(mw)}$. Once we fill up the table, for each $\widetilde{n}=(n_1, \ldots, n_m)$ such that each $n_i$ is an integer in $\left[ n \right]_{0}$, we check if both of the following conditions are true:
\begin{enumerate}
\item $\forall j \in \{1, \ldots, m\}\setminus \{c\}, n_{c}> n_j$, and
\item $\tbl[r, \{\{x\}\}, \widetilde{n}] \leq \BB$.
\end{enumerate}
If we find such a $\widetilde{n}$, we return \textsf{YES}; otherwise, we return \textsf{NO}. This takes $\OO(n^m\cdot \mathsf{poly}(n, m))$ time. We now proceed to describe how to fill up the entries of $\tbl$.
\subsection*{Initialization}
Let $t$ be a leaf in $T$. Then, by our assumption, $X_t=\emptyset$. We set $\tbl[t,\bot,\epsilon]\gets 0$.
\subsection*{Populating the table}
We inductively assume that while filling up $\tbl[t,\calP,\overline{n}]$, all entries of the form $\tbl[t',\cdot, \cdot]$, where $t'$ is a vertex in the subtree of $T$ rooted at $t$, are filled. The update rule of our dynamic programming algorithm depends on the kind of vertex $t$ is. For the rest of this section, we denote a generic partition by $\calp=\{\calp_1, \ldots, \calp_\ell\}$ and a generic tuple by $\overline{n}=(n_{i,j})_{i\in[\ell], j \in [m]}$; if $\ell=0$, $\calp$ and $\overline{n}$ will be understood to be $\bot$ and $\epsilon$ respectively.
\begin{description}
\item[Introduce node.] Let $t$ be an Introduce node, and let vertex $v \in \VV$ of $\GG$ be introduced at $t$. Thus, $t$ has only one child $t'$, $X_t=X_{t'} \cup \{v\}$, $v \notin V_{t'}$, $E_t=E_{t'}$ and $v$ is an isolated vertex in $G_t$.

Fix a partition $\calp=\{\calp_1,\ldots,\calp_\ell\}$ of a subset of $X_t$ and a tuple $\overline{n}=(n_{i,j})_{i\in [\ell], j \in [m]}$. We consider the following cases, depending on $\calp$ and $\overline{n}$.\\
\begin{enumerate}
\item Suppose that there does not exist any part of $\calp$ that contains $v$. Then $\WW \subseteq V_t \setminus \{ x \}$ induces $(\calp, \overline{n})$ on $X_t$ \emph{if and only if} $v \in \WW$ and $\WW \setminus \{v\}$ induces $(\calp, \overline{n})$ on $X_{t'}$. We update
\[\tbl[t, \calp, \overline{n}]\gets \tbl[t', \calp, \overline{n}] + \pi(v).\]
Note that the case $\calp=\bot$ and $\overline{n}=\epsilon$ falls here.
\item Let $\calp$ and $\overline{n}$ satisfy both of the following two conditions:
\begin{enumerate}
\item One of the parts of $\calp$ is $\{v\}$. WLOG assume that $\calp_\ell=\{v\}$.
\item For $j \in [m]$, $n_{\ell,j}=\left\{\begin{array}{ll}1 & \mbox{if $\tau(v)=j$,} \\ 0 & \mbox{otherwise.}\end{array}\right.$
\end{enumerate}
 Let $\calp' = \{\calP_1,\ldots, \calp_{\ell-1}\}$. Let $\overline{n}'\coloneqq (n_{i,j})_{i\in [\ell-1], j \in [m]}$. Then $\WW \subseteq V_t$ induces $(\calp, \overline{n})$ on $X_t$ \emph{if and only if} $v \notin \WW$ and $\WW$ induces $(\calp', \overline{n}')$ on $X_{t'}$. We update
 \[\tbl[t, \calp, \overline{n}]\gets \tbl[t', \calp', \overline{n}'].\]
\item If none of the previous two cases is true, then no $\WW \subseteq V_t$ can induce $(\calp, \overline{n})$ on $X_t$. We update $\tbl[t, \calp, \overline{n}]\gets \infty.$
\end{enumerate}
The updation time is $\mathsf{poly}(n, m)$.
\item[Introduce edge node.] Let $t$ be an Introduce edge node. Thus, $t$ has only one child $t'$, and $X_t=X_{t'}$. Let an edge $\{u, v\}$, for $u, v \in X_{t'}$, be introduced at $t$.

Fix a partition $\calp=\{\calp_1,\ldots,\calp_\ell\}$ of a subset of $X_t$ and a tuple $\overline{n}=(n_{i,j})_{i\in [\ell], j \in [m]}$. Now, for a partition $\calp'=(\calp'_1, \ldots, \calp'_{\ell'})$ of a subset $S$ of $X_{t'}=X_t$ and a tuple $\overline{n}'=(n'_{i, j})_{i \in [\ell'], j \in [m]}$, we define a partition $\lambda(\calp')$ of $S$ and a tuple $\sigma(\overline{n}')$ as follows:
\begin{itemize}
\item If $u, v$ belong to $S$, and they belong to two different parts $\calp_{i_1}$ and $\calp_{i_2}$ of $\calp'$, replace those two parts by their union. Define $\lambda(\calp')$ to be the resulting partition. Obtain $\sigma(\overline{n}')$ from $\overline{n}'$ by replacing the entries $n'_{i_1, j}$ and $n'_{i_2, j}$ by their sum, for each $j \in [m]$.
\item  Otherwise define $\lambda(\calp')\coloneqq \calp'$ and $\sigma(\overline{n}')\coloneqq \overline{n}'$.
\end{itemize}
Note that $\WW \subseteq V_t$ induces $(\calp, \overline{n})$ on $(X_t, G_t)$ if and only if $\WW$ induces $(\calp', \overline{n}')$ on $(X_{t'}, G_{t'})$ for some $\calp'$ and $\overline{n}'$ such that $\lambda(\calp')=\calp$ and $\sigma(\overline{n}')=\overline{n}$. This leads us to our update rule:
\[\tbl[t, \calp, \overline{n}] \gets \min_{\calp', \overline{n}'}\tbl[t', \calp', \overline{n}'],\]
where the minimum is over all pairs $(\calp', \overline{n}')$ such that $\lambda(\calp')=\calp$ and $\sigma(\overline{n}')=\overline{n}$. The updation takes time $w^{\OO(w)}\cdot n^{\OO(mw)}\cdot \mathsf{poly}(n, m)$.

\item[Forget vertex node.] Let $t$ be a Forget vertex node, and let vertex $v\in \VV$ of $\GG$ be forgotten at $t$. Thus, $t$ has only one child $t'$, $X_t=X_{t'}\setminus\{v\}$ and $v \in X_{t'}$. However, note that $G_t$ and $G_{t'}$ are the same graph.

Fix a partition $\calp=\{\calp_1,\ldots,\calp_\ell\}$ of a subset of $X_t$ and a tuple $\overline{n}=(n_{i,j})_{i\in [\ell], j \in [m]}$. Let $\calp'=\{\calp'_1,\ldots,\calp'_{\ell'}\}$ be a partition of a subset of $X_{t'}$. Let $\overline{n}'$ be a tuple $(n'_{i,j})_{i\in [\ell'], j \in [m]}$. We say that the pair $(\calp', \overline{n}')$ is \emph{consistent with} the pair $(\calp, \overline{n})$ if one of the following is true:
\begin{enumerate}
\item $\calp=\calp'$ and $\overline{n}=\overline{n}'$.
\item $\ell=\ell'$ and there exists an $i \in [\ell]$ such that both of the following are true
\begin{enumerate}
\item $\calp'_i=\calp_i \cup \{v\}$ and for each $j \in [m]$, $n'_{i, j}=n_{i, j}$
\item $\forall k \in [\ell] \setminus \{i\}$, $\calp'_k=\calp_k$ 
and for all $j \in [m]$, $n'_{k, j}=n_{k, j}$.
\end{enumerate}
\item $\ell +1 = \ell'$, $\calp'=\calp \cup \{\{v\}\}$. WLOG, assume that $\calp'_{\ell'}=\{v\}$. Furthermore, for all $i \in [\ell], j \in [m]$, $n'_{i,j}=n_{i,j}$. 
\end{enumerate}
We now argue that $\WW\subseteq V_t$ induces $(\calp, \overline{n})$ on $X_t$, \emph{if and only if} $\WW$ induces some $(\calp', \overline{n}')$ consistent with $(\calp, \overline{n})$ on $X_{t'}$. To see this, recall that the graphs $G_t$ and $G_{t'}$ are the same. Thus, the connected components $\{\calY^{\WW,t}_i\}_i$ and $\{\calY^{\WW,t'}_i\}_i$ are the same. Now $X_{t'}=X_t \cup \{v\}$. Thus for each $i$, one of the following is true: 1. $\calP'_i=\calY^{\WW,t'}_i \cap X_{t'}=\{v\}$ and $\calY^{\WW,t}_i \cap X_{t}=\emptyset$, 2. $\calp'_i=\calp_i$, 3. $v \in \calp'_i$ and $\calp'_i\setminus \{v\}=\calp_i$. Furthermore, (2) holds for all but at most one $i$. It can be verified that the tuples $\overline{n}$ and $\overline{n}'$ must be related so as to ensure that the pairs $(\calp', \overline{n}')$ and $(\calp, \overline{n})$ are consistent. This brings us to our update rule:
\[\tbl[t, \calp, \overline{n}]\gets \min_{(\calp', \overline{n}')\text{ consistent with }(\calp, \overline{n})} \tbl[t', \calp', \overline{n}'].\]
The updation takes time $w^{\OO(w)}\cdot n^{\OO(mw)}\cdot \mathsf{poly}(n, m)$.
\item[Join node.] Let $t$ be a Join node. Thus, it has two children $t_1$ and $t_2$. Furthermore, $X_t=X_{t_1}=X_{t_2}$.

It follows from the definition of join node that $X_t = X_{t_1}=X_{t_2} \subseteq V_{t_1} \cap V_{t_2}$. Can there be a vertex $v$ in $V_{t_1} \cap V_{t_2}$ that is not in $X_t$? It is easy to see that the answer is `no', as otherwise the subgraph of $T$ induced by the set $\{t \mid v \in X_t\}$ will not be connected. We conclude that $V_{t_1} \cap V_{t_2}=X_t$.

Fix a partition $\calp=(\calp_1, \ldots, \calp_\ell)$ of a subset of $X_t$, and a tuple $\overline{n}=(n_{i,j})_{i\in[\ell], j \in [m]}$.

Suppose $\WW \subseteq V_t$ induces $(\calp, \overline{n})$ on $(X_t, G_t)$. Express $\WW$ as the disjoint union
\[\WW=\WW_b \uplus\WW_1\uplus \WW_2,\mbox{ where}\]
\begin{itemize}
\item  $\WW_b=\WW \cap X_t$,
\item $\WW_1=(\WW \setminus \WW_b) \cap V_{t_1}$, and
\item $\WW_2=(\WW \setminus \WW_b) \cap V_{t_2}$.
\end{itemize}
Let $\calp'$ and $\calp''$ be the partitions of subsets of $X_{t}$ (which, recall, is the same as $X_{t_1}$ and $X_{t_2}$), and $\overline{\mu}, \overline{\nu}$ be vectors of numbers such that $\WW_b \uplus\WW_1$ induces $(\calp', \overline{\mu})$ on $(X_{t_1}, G_{t_1})$, and $\WW_b \uplus\WW_2$ induces $(\calp'', \overline{\nu})$ on $(X_{t_2}, G_{t_2})$. Can we express $\calP$ and $\overline{n}$ in terms of $\calp', \calp'', \overline{\mu}$ and $\overline{\nu}$? In other words, are $\calP$ and $\overline{n}$ determined by $\calp', \calp'', \overline{\mu}$ and $\overline{\nu}$? We will answer it in the affirmative. Let us start with $\calp$.

 Let $u$ and $v$ be two vertices in $X_t\setminus\WW_b$. We wish to determine whether or not $u$ and $v$ are in the same part of $\calp$.

Assume that there is a path $p$ between $u$ and $v$ in $G_t\setminus \WW$. Note that the edge set of $G_t\setminus\WW$ is the disjoint union (each edge is introduced exactly once) of the edge sets of $G_{t_1}\setminus\left(\WW_b \uplus \WW_1\right)$ and $G_{t_2}\setminus\left(\WW_b \uplus \WW_2\right)$. Also recall that the only vertices common to $V_{t_1}$ and $V_{t_2}$ are the ones in $X_t$. Thus $p$ is necessarily of the form $p_1p_2\ldots p_t$, where each $p_i$ is a path between two vertices in $X_t\setminus\WW_b$ all of whose edges are either in $G_{t_1}\setminus\left(\WW_b \uplus \WW_1\right)$ or in $G_{t_2}\setminus\left(\WW_b \uplus \WW_2\right)$.

To capture this, we define a graph $H=(X_t\setminus \WW_b , E)$ as follows: for each $u, v \in X_t\setminus \WW_b$, the pair $\{u, v\}$ is an edge in $E$ if $u$ and $v$ are in the same part in $\calp'$ or in $\calp''$. It follows that there exists path $p=p_1p_2\ldots p_t$ as stated before if and only if $u$ and $v$ are connected in the above graph. We have the following procedure to determine $\calp$ from $\calp'$ and $\calp''$:
\begin{enumerate}
\item Build the graph $H$ from $\calp'$ and $\calp''$.
\item Return the connected components of $H$ as the parts $\calp_i$s of $\calp$.
\end{enumerate}
The procedure above can be implemented in time polynomial in $|X_t|\leq w$.

Now, let us express $\overline{n}$ in terms of $\overline{\mu}, \overline{\nu}, \calp'$ and $\calp''$. To this end, let
\begin{itemize}
\item $\calY^{(\WW,t)}_1,\ldots,\calY^{(\WW,t)}_\ell$ be the connected components of $G_t\setminus\WW$ with non-empty intersection with $X_t$,
\item $\calY^{(\WW_1,t_1)}_1,\ldots,\calY^{(\WW_1,t_1)}_{\ell_1}$ be the connected components of $G_{t_1}\setminus\left(\WW_b \uplus \WW_1\right)$ with non-empty intersection with $X_{t_1}=X_t$, and
\item $\calY^{(\WW_2,t_2)}_1,\ldots,\calY^{(\WW_2,t_2)}_{\ell_2}$ be the connected components of $G_{t_2}\setminus\left(\WW_b \uplus \WW_2\right)$ with non-empty intersection with $X_{t_2}=X_t$.
\end{itemize}
It is easy to see that if two vertices are connected in $G_{t_1}\setminus\left(\WW_b \uplus \WW_1\right)$ (resp. $G_{t_2}\setminus\left(\WW_b \uplus \WW_2\right)$), then they are connected in $G_t \setminus \WW$ as well. This lets us infer the following:
\begin{enumerate}
\item[($\alpha$)] Each connected component of $G_t\setminus \WW$ is the union of a set (possibly empty) of connected components of $G_{t_1}\setminus \left(\WW_b \uplus \WW_1\right)$ and a set (possibly empty) of connected components of $G_{t_2}\setminus\left(\WW_b \uplus \WW_2\right)$.
\item[($\beta$)] Each part of $\calP$ is the union of a set of parts of $\calp'$ (resp. $\calp''$). Put another way, $\calp'$ and $\calp''$ are refinements of $\calp$.
\item[($\gamma$)] Each $\calY^{(\WW_1,t_1)}_i$ (resp. $\calY^{(\WW_2,t_2)}_i)$ is completely contained in some $\calY^{(\WW,t)}_j$.
\end{enumerate}
From the above, we infer the following. Suppose $\calp_i=\cup_{j \in S_1} (\calp')_j=\cup_{k \in S_2} (\calp'')_k$. Then, $\bigcup_{j \in S_1}\calY^{(\WW_1,t_1)}_j \bigcup_{k \in S_2} \calY^{(\WW_2,t_2)}_k \subseteq \calY^{(\WW,t)}_i$.

Can $\bigcup_{j \in S_1}\calY^{(\WW_1,t_1)}_j \bigcup_{k \in S_2} \calY^{(\WW_2,t_2)}_k$ be a proper subset of $\calY^{(\WW,t)}_i$? In view of the point ($\alpha$) above, the question boils down to the question whether there can exist a connected component of $G_{t_1}\setminus\left(\WW_b \uplus \WW_1\right)$ or $G_{t_2}\setminus\left(\WW_b \uplus \WW_2\right)$, which (i) has empty intersection with $X_t$, and (ii) is completely contained in $\calY^{(\WW,t)}_i$.

We claim that the answer to the above question is `no.' To prove the claim, towards a contradiction, assume that there exists a connected component $K$ of $G_{t_1}\setminus\left(\WW_b \uplus \WW_1\right)$ that satisfies the conditions (i) and (ii) above. The case of there being such a connected component of $G_{t_2}\setminus\left(\WW_b \uplus \WW_2\right)$ is similar.

Since $K \subseteq \calY^{(\WW,t)}_i$ and $\calY^{(\WW,t)}_i \cap X_t \neq \emptyset$, we have that for each vertex $u \in K$ there is a path $p$ to a vertex $w \in X_t$ in $G_t \setminus \WW$. Since $E_t$ is a disjoint union of $E_{t_1}$ and $E_{t_2}$ as noted before, and $V_{t_1} \cap V_{t_2}=X_t$, we have that $p$ is of the form $p_1\ldots p_{h}$, where each $p_i$ is a path either in $G_{t_1} \setminus\left(\WW_b \uplus \WW_1\right)$ or in $G_{t_2} \setminus\left(\WW_b \uplus \WW_2\right)$ which ends at a vertex in $X_t$. Since $u$ is a vertex of $G_{t_1} \setminus\left(\WW_b \uplus \WW_1\right)$, we have that $p_1$ is a path in $G_{t_1} \setminus\left(\WW_b \uplus \WW_1\right)$ connecting $u$ with a vertex in $X_t$, which contradicts the assumption that $K$ does not intersect $X_t$. This proves the claim. We conclude that $\bigcup_{j \in S_1}\calY^{(\WW_1,t_1)}_j \bigcup_{k \in S_2} \calY^{(\WW_2,t_2)}_k = \calY^{(\WW,t)}_i$.

Now we are ready to express $\overline{n}$ in terms of $\overline{\mu}, \overline{\nu}, \calp'$ and $\calp''$. Since $V_{t_1}\setminus X_t$ and $V_{t_2} \setminus X_t$ are disjoint, we have that for each $i \in [\ell], j \in [m]$, $n_{i, j}=\sum_{u \in S_1}\mu_{u, j}+\sum_{v \in S_2}\nu_{v, j}-N_{i, j}$, where $N_{i, j}$ is the number of nodes in $\calp_i$ that votes for $j$.

We have described a procedure $\mathcal A$ that, given $\calp', \calp'', \overline{\mu}$ and $\overline{\nu}$, produces $\calp$ and $\overline{n}$ in time polynomial in $m$ and $n$, with the following property: $\WW$ induces $(\calp, \overline{n})$ on $(X_t, G_t)$ \emph{if and only if} there exists $\calp', \calp'', \overline{\mu}$ and $\overline{\nu}$ such that 1. $\WW_b \uplus \WW_1$ induces $(\calp', \overline{\mu})$ on $(X_{t_1}, G_{t_1})$, 2. $\WW_b \uplus \WW_2$ induces $(\calp'', \overline{\nu})$ on $(X_{t_2}, G_{t_2})$, and 3. $\mathcal{A}(\calp', \calp'', \overline{\mu}, \overline{\nu}) = (\calp, \overline{n})$. Observe that $\pi(\WW)=\pi(\WW_b \uplus \WW_1)+\pi(\WW_b \uplus \WW_2)-\pi(X_t\setminus \cup_{i\in[\ell]}\calp_i)$. We thus have the following update rule:
\[\tbl[t, \calp, \overline{n}]\gets \min_{\calp', \calp'', \overline{\mu}, \overline{\nu}}\tbl[t_1, \calp', \overline{\mu}]+\tbl[t_2, \calp'', \overline{\nu}]-\pi(X_t\setminus \cup_{i\in[\ell]}\calp_i),\]
where the minimum is over all tuples $(\calp', \calp'', \overline{\mu}, \overline{\nu})$ such that $\mathcal{A}(\calp', \calp'', \overline{\mu}, \overline{\nu})=(\calp, \overline{n})$. Given $(\calp', \overline{\mu})$ and $(\calp'', \overline{\nu})$, one can find out $\calp$ and $\overline{n}$ by going over all partitions $\calp', \calp''$ and vectors $\overline{\mu}, \overline{\nu}$; this takes time $w^{\OO(w)}\cdot n^{\OO(mw)} \cdot \mathsf{poly}(n, m)$.
\end{description}
\end{proof}

\section{Algorithm for \DCN}
In this section, we prove \Cref{thm:easy-low-tw-dcn}. The broad idea is to check for each candidate $c'$ other than $c$, whether $c'$ can be made to garner strictly more votes than $c$ within the controller's budget. The aforementioned task for each candidate is accomplished by a dynamic programming algorithm, which closely follows the one for the constructive case (\Cref{algo:cn}). We assume that the tree decomposition $T$ of $\GG$ supplied as input is \emph{nice} (see \Cref{sec:tw}). Similar to the constructive case, our DP table has a set of cells for each node $t$ of $T$. Each cell is additionally indexed by a partition $\calp=(\calp_1, \ldots, \calp_\ell)$ of a subset of the bag $X_t$ of $t$, and a tuple of numbers $\overline{n}=(n_{i})_{i \in [\ell]}$. The vertices and edges of $\GG$, introduced in the subtree of $T$ rooted at $t$, induce a subgraph $G_t$ of $\GG$. In each step, the algorithm computes a cheapest set of vertices of $G_t$. Removal of such a cheapest set of vertices from $G_t$ breaks up $G_t$ into connected components whose intersections with the bag $X_t$ are given by $\calp$, and for each such components, the vote count differences between $c^{\prime}$ and $c$ are given by $\overline{n}$. Once the table is filled up, the table entries for the root of $T$ can be consulted to compute the output of the algorithm. We now proceed to the formal proof.

\begin{proof}[Proof of \Cref{thm:easy-low-tw-dcn}]
For every candidate $c'\in \CC \setminus \{c\}$ we determine whether there is a subset $\WW \subseteq \VV \setminus \{ x \}$ of total cost at most $\BB$ such that $c'$ garners more votes than $c$ in the restricted election where only voters reachable from $x$ in $\GG \setminus \WW$ are counted. Let \CB($c') \in\{\textsf{YES}, \textsf{NO}\}$ denote whether or not there exists such a $\WW$. Our overall algorithm will output \textsf{YES} if we find a candidate $c'$ for which \CB($c')=\textsf{YES}$, and output \textsf{NO} otherwise. Towards this plan, let us now fix a candidate $c' \neq c$.

We next propose a dynamic programming algorithm to determine \CB($c')$.
Without loss of generality, assume $T$ to be a nice tree decomposition of $\GG$ of width $w$. Let $r$ be the root of $T$. We assume that the parent of each leaf is an \emph{Introduce node} where the vertex $x$ is introduced and that $x$ is never forgotten. In particular, $x$ is in the bag of the root of $T$. See the related discussion in \Cref{sec:prelims}.
Consider a vertex $t \in V_T$. Let $V_t$ be the union of the bags of all vertices in the subtree of $T$ rooted at $t$. Let $E_t$ be the set of all edges introduced in the subtree of $T$ rooted at $t$. Define the graph $G_t\coloneqq (V_t, E_t)$.

Consider deleting a subset $\WW \subseteq V_t \setminus \{x\}$ of vertices from $G_t$. The resulting graph, which we will denote by $G_t \setminus \WW$, will have its vertex set partitioned into connected components. As in the Constructive case, let $\ell$ of those components, denoted by $\calY^{(\WW,t)}_1,\ldots,\calY^{(\WW,t)}_\ell$, have non-empty intersections with the bag $X_t$ of $t$. For $i=1,\ldots,\ell$ define $\calP^{(\WW,t)}_i\coloneqq \calY^{(\WW,t)}_i\cap X_t$. Let $\calP^{(\WW,t)}\coloneqq (\calP^{(\WW,t)}_1,\ldots,\calP^{(\WW,t)}_\ell)$ be the partition of $X_t\setminus \WW$ defined by the $\calp^{(\WW,t)}_i$s. By definition, each $\calP^{(\WW,t)}_i$ is non-empty, which implies that $\ell \leq |X_t| \leq w+1$.

We now describe the entries of our dynamic programming table, $\tbl$. Each entry is indexed by a tuple $(t, \calp, \overline{n})$, whose components are described below.
\begin{itemize}
\item $t \in V_T$ is a vertex of $T$.
\item $\calP=(\calP_1,\ldots,\calp_\ell)$ is a partition of a subset $S$ of $X_t$ into non-empty parts $\calp_i$. If $S=\emptyset$, then $\calP$ is the empty partition that we denote by $\bot$.
\item If $S\neq \emptyset$, then $\overline{n}$ is a tuple $(n_{i})_{i\in [\ell]}$ of $\ell$ numbers, where each $n_{i}$ is an integer in $[-n, n]$. If $S=\emptyset$, define $\overline{n}$ to be the empty sequence that we denote by $\epsilon$.
\end{itemize}
After our algorithm terminates, $\tbl[t, \calp, \overline{n}]$ will contain the minimum cost of any subset $\WW \subseteq V_t \setminus \{x\}$ which satisfies the following conditions:
\begin{enumerate}
\item $\calP^{(\WW, t)}=\calp$, and
\item For each $i \in [\ell]$, if $a_c$ and $a_{c'}$ be the numbers of votes garnered by $c$ and $c'$ respectively in $\calY^{(\WW, t)}_i$, then $n_{i}=a_{c'} - a_c$.
\end{enumerate}
We say that a $\WW$ that satisfies the above two conditions induces the pair $(\calP, \overline{n})$ on $(X_t, G_t)$. If there is no such $\WW$, $\tbl[t, \calp, \overline{n}]$ will be $\infty$.

The number of entries in our table is at most $w^{\OO(w)}\cdot n^{\OO(w)}$. Once we fill up the table, for each $a>0$ we check whether $\tbl[r, \{\{x\}\}, a] \leq \BB$. If we find such an $a$, we return \textsf{YES} as the answer to $\CB$; otherwise, we return \textsf{NO}. This takes $ \mathsf{poly}(n)$ time. We now proceed to describe how to fill up the entries of $\tbl$.
\subsection*{Initialization}
Let $t$ be a leaf in $T$. Then, by our assumption, $X_t=\emptyset$. Set $\tbl[t,\bot,\epsilon]\gets 0$.
\subsection*{Populating the table} 

Similar to the constructive case, we inductively assume that while filling up $\tbl[t,\calP,\overline{n}]$, all entries of the form $\tbl[t',\cdot, \cdot]$, where $t'$ is a vertex in the subtree of $T$ rooted at $t$, are filled. The update rule of our dynamic programming algorithm depends on the kind of vertex $t$ is. For the rest of this section, we denote a generic partition by $\calp=\{\calp_1, \ldots, \calp_\ell\}$ and a generic tuple by $\overline{n}=(n_{i})_{i\in[\ell]}$; if $\ell=0$, then $\calp$ and $\overline{n}$ will be understood to be $\bot$ and $\epsilon$ respectively.
\begin{description}
\item[Introduce node.] Let $t$ be an Introduce node, and let vertex $v \in \VV$ of $\GG$ be introduced at $t$. Thus, $t$ has only one child $t'$, $X_t=X_{t'} \cup \{v\}$, $v \notin V_{t'}$, $E_t=E_{t'}$ and $v$ is an isolated vertex in $G_t$.

Fix a partition $\calp=\{\calp_1,\ldots,\calp_\ell\}$ of a subset of $X_t$ and a tuple $\overline{n}=(n_{i})_{i\in [\ell]}$. We consider the following cases, depending on $\calp$ and $\overline{n}$.\\
\begin{enumerate}
\item Suppose that there does not exist any part of $\calp$ that contains $v$. Then $\WW \subseteq V_t \setminus \{ x \}$ induces $(\calp, \overline{n})$ on $X_t$ \emph{if and only if} $v \in \WW$ and $\WW \setminus \{v\}$ induces $(\calp, \overline{n})$ on $X_{t'}$. We update
\[\tbl[t, \calp, \overline{n}]\gets \tbl[t', \calp, \overline{n}] + \pi(v).\]
Note that the case $\calp=\bot$ and $\overline{n}=\epsilon$ falls here.
\item Let $\calp$ and $\overline{n}$ satisfy both of the following two conditions:
\begin{enumerate}
\item One of the parts of $\calp$ is $\{v\}$. WLOG assume that $\calp_\ell=\{v\}$.
\item $n_{\ell}=\left\{\begin{array}{lll}1 & \mbox{if $\tau(v)=c'$,} \\ -1 & \mbox{if $\tau(v)=c$, and} \\ 0 & \mbox{otherwise.}\end{array}\right.$
\end{enumerate}
 Let $\calp' = \{\calP_1,\ldots, \calp_{\ell-1}\}$ and $\overline{n}'\coloneqq (n_{i})_{i\in [\ell-1]}$. Then $\WW \subseteq V_t$ induces $(\calp, \overline{n})$ on $X_t$ \emph{if and only if} $v \notin \WW$ and $\WW$ induces $(\calp', \overline{n}')$ on $X_{t'}$. We update
 \[\tbl[t, \calp, \overline{n}]\gets \tbl[t', \calp', \overline{n}'].\]
\item If none of the previous two cases is true, then no $\WW \subseteq V_t$ can induce $(\calp, \overline{n})$ on ($X_t, G_t)$. We update $\tbl[t, \calp, \overline{n}]\gets \infty$.
\end{enumerate}
The updation time is $\mathsf{poly}(n, m)$.
\item[Introduce edge node.] Let $t$ be an Introduce edge node. Thus, $t$ has only one child $t'$, and $X_t=X_{t'}$. Let an edge $\{u, v\}$, for $u, v \in X_{t'}$, be introduced at $t$.

Fix a partition $\calp=\{\calp_1,\ldots,\calp_\ell\}$ of a subset of $X_t$ and a tuple $\overline{n}=(n_{i})_{i\in [\ell]}$. Now, for a partition $\calp'=(\calp'_1, \ldots, \calp'_{\ell'})$ of a subset $S$ of $X_{t'}=X_t$ and a tuple $\overline{n}'=(n'_{i})_{i \in [\ell']}$, we define a partition $\lambda(\calp')$ of $S$ and a tuple $\sigma(\overline{n}')$ as follows:
\begin{itemize}
\item If $u, v$ belong to $S$, and they belong to two different parts of $\calp'$, replace those two parts by their union. Define $\lambda(\calp')$ to be the resulting partition. Obtain $\sigma(\overline{n}')$ from $\overline{n}'$ by replacing the entries $n'_{i_1}$ and $n'_{i_2}$ by their sum, for each $j \in [m]$.
\item  Otherwise define $\lambda(\calp')\coloneqq \calp'$ and $\sigma(\overline{n}')\coloneqq \overline{n}'$.
\end{itemize}
Note that $\WW \subseteq V_t$ induces $(\calp, \overline{n})$ on $(X_t, G_t)$ if and only if $\WW$ induces $(\calp', \overline{n}')$ on $(X_{t'}, G_{t'})$ for some $\calp'$ and $\overline{n}'$ such that $\lambda(\calp')=\calp$ and $\sigma(\overline{n}')=\overline{n}$. This leads us to our update rule:
\[\tbl[t, \calp, \overline{n}] \gets \min_{\calp', \overline{n}'}\tbl[t', \calp', \overline{n}'],\]
where the minimum is over all pairs $(\calp', \overline{n}')$ such that $\lambda(\calp')=\calp$ and $\sigma(\overline{n}')=\overline{n}$. The updation takes time $w^{\OO(w)}\cdot n^{\OO(w)}\cdot \mathsf{poly}(n)$.

\item[Forget vertex node.] Let $t$ be a Forget vertex node, and let vertex $v\in \VV$ of $\GG$ be forgotten at $t$. Thus, $t$ has only one child $t'$, $X_t=X_{t'}\setminus\{v\}$ and $v \in X_{t'}$. However, note that $G_t$ and $G_{t'}$ are the same graph.

Fix a partition $\calp=\{\calp_1,\ldots,\calp_\ell\}$ of a subset of $X_t$ and a tuple $\overline{n}=(n_{i})_{i\in [\ell]}$. Let $\calp'=\{\calp'_1,\ldots,\calp'_{\ell'}\}$ be a partition of a subset of $X_{t'}$. Let $\overline{n}'$ be a tuple $(n'_{i})_{i\in [\ell']}$. We say that the pair $(\calp', \overline{n}')$ is \emph{consistent with} the pair $(\calp, \overline{n})$ if one of the following is true:
\begin{enumerate}
\item $\calp=\calp'$ and $\overline{n}=\overline{n}'$.
\item $\ell=\ell'$ and there exists an $i \in [\ell]$ such that both of the following are true
\begin{enumerate}
\item $\calp'_i=\calp_i \cup \{v\}$ and \\$n_{i}=\left\{\begin{array}{lll}n'_{i}-1 & \mbox{if $\tau(v)=c'$,}\\ n'_{i}+1 & \mbox{if $\tau(v)=c$, and}\\ n'_{i} & \mbox{otherwise.}\end{array}\right.$
\item $\forall k \in [\ell] \setminus \{i\}$, $\calp'_k=\calp_k$ and $n'_{k}=n_{k}$.
\end{enumerate}
\item $\ell +1 = \ell'$, $\calp'=\calp \cup \{\{v\}\}$. WLOG, assume that $\calp'_{\ell'}=\{v\}$. Furthermore, for all $i \in [\ell]$, $n'_{i}=n_{i}$. 
\end{enumerate}
We now argue that $\WW\subseteq V_t$ induces $(\calp, \overline{n})$ on $(X_t, G_t)$, \emph{if and only if} $\WW$ induces some $(\calp', \overline{n}')$ consistent with $(\calp, \overline{n})$ on $(X_{t'}, G_{t'})$. To see this, recall that the graphs $G_t$ and $G_{t'}$ are the same. Thus, the connected components $\{\calY^{\WW,t}_i\}_i$ and $\{\calY^{\WW,t'}_i\}_i$ are the same. Now $X_{t'}=X_t \cup \{v\}$. Thus for each $i$, one of the following is true: 1. $\calP'_i=\calY^{\WW,t'}_i \cap X_{t'}=\{v\}$ and $\calY^{\WW,t}_i \cap X_{t}=\emptyset$, 2. $\calp'_i=\calp_i$, 3. $v \in \calp'_i$ and $\calp'_i\setminus \{v\}=\calp_i$. Furthermore, (2) holds for all but at most one $i$. It can be verified that the tuples $\overline{n}$ and $\overline{n}'$ must be related so as to ensure that the pairs $(\calp', \overline{n}')$ and $(\calp, \overline{n})$ are consistent.

This brings us to our update rule:
\[\tbl[t, \calp, \overline{n}]\gets \min_{(\calp', \overline{n}')\text{ consistent with }(\calp, \overline{n})} \tbl[t', \calp', \overline{n}'].\]
The updation takes time $w^{\OO(w)}\cdot n^{\OO(w)}\cdot \mathsf{poly}(n)$.
\item[Join node.] Let $t$ be a Join node. Thus, it has two children $t_1$ and $t_2$. Furthermore, $X_t=X_{t_1}=X_{t_2}$.

It follows from the definition of join node that $X_t = X_{t_1}=X_{t_2} \subseteq V_{t_1} \cap V_{t_2}$. Can there be a vertex $v$ in $V_{t_1} \cap V_{t_2}$ that is not in $X_t$? It is easy to see that the answer is `no', as otherwise the subgraph of $T$ induced by the set $\{t \mid v \in X_t\}$ will not be connected. We conclude that $V_{t_1} \cap V_{t_2}=X_t$.

Fix a partition $\calp=(\calp_1, \ldots, \calp_\ell)$ of a subset of $X_t$, and a tuple $\overline{n}=(n_{i})_{i\in[\ell]}$.

Suppose $\WW \subseteq V_t$ induces $(\calp, \overline{n})$ on $(X_t, G_t)$. Express $\WW$ as the disjoint union
\[\WW=\WW_b \uplus\WW_1\uplus \WW_2,\mbox{ where}\]
\begin{itemize}
\item  $\WW_b=\WW \cap X_t$,
\item $\WW_1=(\WW \setminus \WW_b) \cap V_{t_1}$, and
\item $\WW_2=(\WW \setminus \WW_b) \cap V_{t_2}$.
\end{itemize}
Let $\calp'$ and $\calp''$ be the partitions of subsets of $X_{t}$ (which, recall, is the same as $X_{t_1}$ and $X_{t_2}$), and $\overline{\mu}, \overline{\nu}$ be vectors of numbers such that $\WW_b \uplus\WW_1$ induces $(\calp', \overline{\mu})$ on $(X_{t_1}, G_{t_1})$, and $\WW_b \uplus\WW_2$ induces $(\calp'', \overline{\nu})$ on $(X_{t_2}, G_{t_2})$. Can we express $\calP$ and $\overline{n}$ in terms of $\calp', \calp'', \overline{\mu}$ and $\overline{\nu}$? In other words, are $\calP$ and $\overline{n}$ determined by $\calp', \calp'', \overline{\mu}$ and $\overline{\nu}$? We will answer it in the affirmative. Let us start with $\calp$.

 Let $u$ and $v$ be two vertices in $X_t\setminus\WW_b$. We wish to determine whether or not $u$ and $v$ are in the same part of $\calp$.

Assume that there is a path $p$ between $u$ and $v$ in $G_t\setminus \WW$. Note that the edge set of $G_t\setminus\WW$ is the disjoint union (each edge is introduced exactly once) of the edge sets of $G_{t_1}\setminus\left(\WW_b \uplus \WW_1\right)$ and $G_{t_2}\setminus\left(\WW_b \uplus \WW_2\right)$. Also recall that the only vertices common to $V_{t_1}$ and $V_{t_2}$ are the ones in $X_t$. Thus $p$ is necessarily of the form $p_1p_2\ldots p_t$, where each $p_i$ is a path between two vertices in $X_t\setminus\WW_b$ all of whose edges are either in $G_{t_1}\setminus\left(\WW_b \uplus \WW_1\right)$ or in $G_{t_2}\setminus\left(\WW_b \uplus \WW_2\right)$.

To capture this, we define a graph $H=(X_t\setminus \WW_b , E)$ as follows: for each $u, v \in X_t\setminus \WW_b$, the pair $\{u, v\}$ is an edge in $E$ if $u$ and $v$ are in the same part in $\calp'$ or in $\calp''$. It follows that there exists path $p=p_1p_2\ldots p_t$ as stated before if and only if $u$ and $v$ are connected in the above graph. We have the following procedure to determine $\calp$ from $\calp'$ and $\calp''$:
\begin{enumerate}
\item Build the graph $H$ from $\calp'$ and $\calp''$.
\item Return the connected components of $H$ as the parts $\calp_i$s of $\calp$.
\end{enumerate}
The procedure above can be implemented in polynomial time.

Now, let us express $\overline{n}$ in terms of $\overline{\mu}, \overline{\nu}, \calp'$ and $\calp''$. To this end, let
\begin{itemize}
\item $\calY^{(\WW,t)}_1,\ldots,\calY^{(\WW,t)}_\ell$ be the connected components of $G_t\setminus\WW$ with non-empty intersection with $X_t$,
\item $\calY^{(\WW_1,t_1)}_1,\ldots,\calY^{(\WW_1,t_1)}_{\ell_1}$ be the connected components of $G_{t_1}\setminus\left(\WW_b \uplus \WW_1\right)$ with non-empty intersection with $X_{t_1}=X_t$, and
\item $\calY^{(\WW_2,t_2)}_1,\ldots,\calY^{(\WW_2,t_2)}_{\ell_2}$ be the connected components of $G_{t_2}\setminus\left(\WW_b \uplus \WW_2\right)$ with non-empty intersection with $X_{t_2}=X_t$.
\end{itemize}
It is easy to see that if two vertices are connected in $G_{t_1}\setminus\left(\WW_b \uplus \WW_1\right)$ (resp. $G_{t_2}\setminus\left(\WW_b \uplus \WW_2\right)$), then they are connected in $G_t \setminus \WW$ as well. This lets us infer the following:
\begin{enumerate}
\item[($\alpha$)] Each connected component of $G_t\setminus \WW$ is the union of a set (possibly empty) of connected components of $G_{t_1}\setminus \left(\WW_b \uplus \WW_1\right)$ and a set (possibly empty) of connected components of $G_{t_2}\setminus\left(\WW_b \uplus \WW_2\right)$.
\item[($\beta$)] Each part of $\calP$ is the union of a set of parts of $\calp'$ (resp. $\calp''$). Put another way, $\calp'$ and $\calp''$ are refinements of $\calp$.
\item[($\gamma$)] Each $\calY^{(\WW_1,t_1)}_i$ (resp. $\calY^{(\WW_2,t_2)}_i)$ is completely contained in some $\calY^{(\WW,t)}_j$.
\end{enumerate}
From the above, we infer the following. Suppose $\calp_i=\cup_{j \in S_1} (\calp')_j=\cup_{k \in S_2} (\calp'')_k$. Then, $\bigcup_{j \in S_1}\calY^{(\WW_1,t_1)}_j \bigcup_{k \in S_2} \calY^{(\WW_2,t_2)}_k \subseteq \calY^{(\WW,t)}_i$.

Can $\bigcup_{j \in S_1}\calY^{(\WW_1,t_1)}_j \bigcup_{k \in S_2} \calY^{(\WW_2,t_2)}_k$ be a proper subset of $\calY^{(\WW,t)}_i$? In view of the point ($\alpha$) above, the question boils down to the question whether there can exist a connected component of $G_{t_1}\setminus\left(\WW_b \uplus \WW_1\right)$ or $G_{t_2}\setminus\left(\WW_b \uplus \WW_2\right)$, which (i) has empty intersection with $X_t$, and (ii) is completely contained in $\calY^{(\WW,t)}_i$.

We claim that the answer to the above question is `no.' To prove the claim, towards a contradiction, assume that there exists a connected component $K$ of $G_{t_1}\setminus\left(\WW_b \uplus \WW_1\right)$ that satisfies the conditions (i) and (ii) above. The case of there being such a connected component of $G_{t_2}\setminus\left(\WW_b \uplus \WW_2\right)$ is similar.

Since $K \subseteq \calY^{(\WW,t)}_i$ and $\calY^{(\WW,t)}_i \cap X_t \neq \emptyset$, we have that for each vertex $u \in K$ there is a path $p$ to a vertex $w \in X_t$ in $G_t \setminus \WW$. Since $E_t$ is a disjoint union of $E_{t_1}$ and $E_{t_2}$ as noted before, and $V_{t_1} \cap V_{t_2}=X_t$, we have that $p$ is of the form $p_1\ldots p_{h}$, where each $p_i$ is a path either in $G_{t_1} \setminus\left(\WW_b \uplus \WW_1\right)$ or in $G_{t_2} \setminus\left(\WW_b \uplus \WW_2\right)$ which ends at a vertex in $X_t$. Since $u$ is a vertex of $G_{t_1} \setminus\left(\WW_b \uplus \WW_1\right)$, we have that $p_1$ is a path in $G_{t_1} \setminus\left(\WW_b \uplus \WW_1\right)$ connecting $u$ with a vertex in $X_t$, which contradicts the assumption that $K$ does not intersect $X_t$. This proves the claim. We conclude that $\bigcup_{j \in S_1}\calY^{(\WW_1,t_1)}_j \bigcup_{k \in S_2} \calY^{(\WW_2,t_2)}_k = \calY^{(\WW,t)}_i$.

Now we are ready to express $\overline{n}$ in terms of $\overline{\mu}, \overline{\nu}, \calp'$ and $\calp''$. Since $V_{t_1}\setminus X_t$ and $V_{t_2} \setminus X_t$ are disjoint, we have that for each $i \in [\ell]$, $n_{i}=\sum_{u \in S_1}\mu_{u}+\sum_{v \in S_2}\nu_{v}-N_{i}$, where $N_{i}$ is the difference between the number of nodes in $\calp_i$ that vote for $c'$ and the number of nodes in $\calp_i$ that vote for $c$.

We have described a procedure $\mathcal A$ that, given $\calp', \calp'', \overline{\mu}$ and $\overline{\nu}$, produces $\calp$ and $\overline{n}$ in time polynomial in $n$, with the following property: $\WW$ induces $(\calp, \overline{n})$ on $(X_t, G_t)$ \emph{if and only if} there exists $\calp', \calp'', \overline{\mu}$ and $\overline{\nu}$ such that 1. $\WW_b \uplus \WW_1$ induces $(\calp', \overline{\mu})$ on $(X_{t_1}, G_{t_1})$, 2. $\WW_b \uplus \WW_2$ induces $(\calp'', \overline{\nu})$ on $(X_{t_2}, G_{t_2})$, and 3. $\mathcal{A}(\calp', \calp'', \overline{\mu}, \overline{\nu}) = (\calp, \overline{n})$. Observe that $\pi(\WW)=\pi(\WW_b \uplus \WW_1)+\pi(\WW_b \uplus \WW_2)-\pi(X_t\setminus \cup_{i\in[\ell]}\calp_i)$. We thus have the following update rule:
\[\tbl[t, \calp, \overline{n}]\gets \min_{\calp', \calp'', \overline{\mu}, \overline{\nu}}\tbl[t_1, \calp', \overline{\mu}]+\tbl[t_2, \calp'', \overline{\nu}]-\pi(X_t\setminus \cup_{i\in[\ell]}\calp_i),\]
where the minimum is over all tuples $(\calp', \calp'', \overline{\mu}, \overline{\nu})$ such that $\mathcal{A}(\calp', \calp'', \overline{\mu}, \overline{\nu})=(\calp, \overline{n})$. Given $(\calp', \overline{\mu})$ and $(\calp'', \overline{\nu})$, one can find out $\calp$ and $\overline{n}$ by going over all partitions $\calp', \calp''$ and vectors $\overline{\mu}, \overline{\nu}$; this takes time $w^{\OO(w)}\cdot n^{\OO(w)} \cdot \mathsf{poly}(n)$.
\end{description}
\end{proof}

\section{Hardness results for $2$ candidates}\label{sec:2cand}
In this section we prove \Cref{thm:hardness-2cand}.
\begin{proof}[Proof of \Cref{thm:hardness-2cand}]
As discussed in \Cref{sec:contrib}, it is sufficient to show that \BCN is NP-complete for $m=2$. We will reduce \XThC (See \Cref{dxef:xthc}) to \BCN.
 \subsection*{The reduction}
 In this section, we describe the reduction. Given an instance $(\ell, S_1, \ldots, S_m)$ of \XThC, we will construct an instance $(\CC,\VV,\tau,\GG,c,x)$ of \BCN, where the
 set $\CC$ has two candidates $0$ and $1$ (i.e.~$\CC=\{0, 1\}$), and the preferred candidate $c$ is the candidate $1$.

 Now, we construct the set $\VV$ of voters as follows. We first define the following pairwise disjoint sets of voters: $\VV_1\coloneqq \{u_1, \ldots, u_{\ell (3m - 1)}\}, R=\{r\}, \VV_2\coloneqq \{v_1, \ldots, v_m\}, \VV_3\coloneqq \{w_{i, j} \mid i \in [m], j \in [3\ell]\}$. We define $\VV\coloneqq \VV_1 \cup R \cup \VV_2 \cup \VV_3$. We define $x\coloneqq u_1$. Next, we define $\tau$.
 \begin{enumerate}
\item For each voter $z \in \VV_1$, $\tau(z)=0$.
\item $\tau(r)=1$.
\item For each voter $z \in \VV_2$, $\tau(z)=0$.
\item For each voter $z \in \VV_3$, $\tau(z)=1$.
 \end{enumerate}
 We are now left with the task of defining the edge set $\EE$ of the graph $\GG=(\VV, \EE)$.
 \begin{enumerate}
 \item For each $i \in [\ell(3m-1)-1]$, there is an edge $\{u_i, u_{i+1}\}$ in $\EE$.
 \item There is the edge $\{u_{\ell(3m-1)}, r\}$ in $\EE$.
 \item For each $i \in [m]$, there is an edge $\{r, v_i\}$ in $\EE$.
 \item For each $i, j \in [m]$ and $k \in [3\ell]$, there is an edge $\{v_i, w_{j, k}\}$ if and only if $k \in S_i$. Thus, if $k \in S_i$, then all the edges $\{v_i, w_{1, k}\}, \ldots, \{v_i, w_{m, k}\}$ are in $\EE$, and if $k \notin S_i$ then none of these edges are in $\EE$.
\end{enumerate}
 As mentioned before, the preferred candidate $c$ is the candidate $1$. This concludes the definition of the instance $(\CC,\VV,\tau,\GG,c,x)$ of \BCN that the reduction produces. It is clear that the reduction runs in polynomial time. Refer to \Cref{R3X_to_G} for a graphical illustration of the reduction stated above.
\begin{figure}[!htbp]
    \centering
    \includegraphics[scale=0.6]{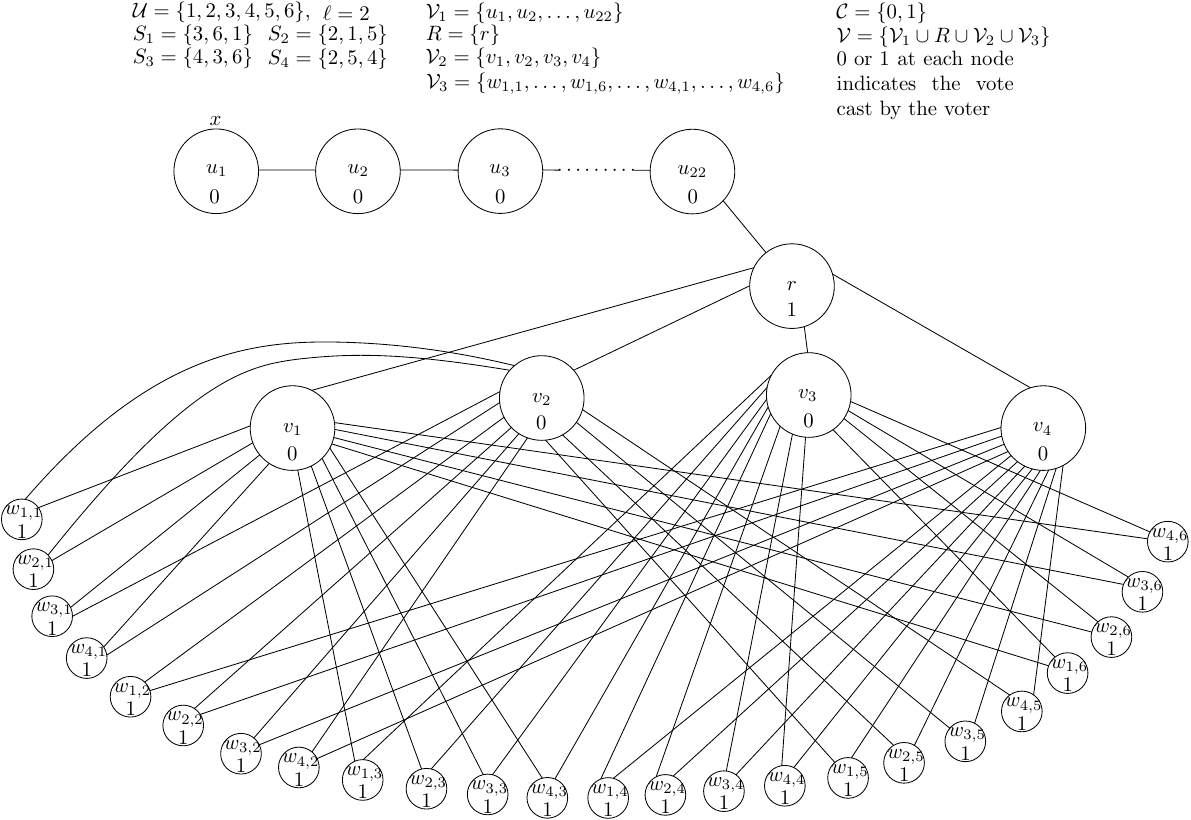}
    \caption{Reduction of \BCN from \XThC}
    \label{R3X_to_G}
\end{figure}
\subsection*{Correctness of the reduction}
Let $(\ell, S_1, \ldots, S_m)$ and $(\CC,\VV,\tau,\GG,c,x)$ be as defined in the previous section. We now show that $(\ell, S_1, \ldots, S_m)$ is a \emph{yes} instance of \XThC \emph{if and only if}
 $(\CC,\VV,\tau,\GG,c,x)$ is a \emph{yes} instance of \BCN.
 \begin{description}
 \item[(If)] Note that if $(\ell, S_1, \ldots, S_m)$ is a valid instance of \XThC then $3m \geq |\cup_{i\in[m]}S_i|=|\UU|=3\ell$ giving us that $m \geq \ell$. Let $(\CC,\VV,\tau,\GG,c,x)$ be a \emph{yes} instance of \BCN. Then, there exists a set $\WW \subseteq \VV \setminus \{u_1\}$ such that the candidate $1$ uniquely wins in the election restricted to the nodes reachable from $u_1$ in $\GG \setminus \WW$. Note that no vertex in $\VV_1 \cup R$ belongs to $\WW$, as otherwise no vertex $v$ with $\tau(v)=1$ will be reachable from $u_1$ in $\GG \setminus \WW$. We may also assume without loss of generality that no vertex in $\VV_3$ belongs to $\WW$, as they all vote for $1$, and their deletion does not serve to disconnect any vertex $v$ with $\tau(v)=0$ from $u_1$.

 We now count the number of votes gathered by candidate $0$ in the restricted election. Define $B=\{i \in [m] \mid v_i \in \WW\}$. Note that all vertices in $\VV_2 \setminus \WW$ are reachable from $u_1$ in $\GG \setminus \WW$, and hence they together contribute $m-|B|$ to the number of votes to candidate $0$. Thus, the total number of votes garnered by $0$ in the election restricted to $\GG \setminus \WW$ is $|\VV_1|+m-|B|=3\ell m - \ell+m-|B|$.

 Now let us count the number of votes to candidate $1$. Vertex $r$ contributes $1$ to the count. Now $w_{j,k} \in \VV_3$ is reachable from $u_1$ in $\GG\setminus \WW$ if and only if there is a set $S_i$ such that (1) $k \in S_i$ and (2) $i \notin B$ (i.e.~ $v_i \notin \WW$). The number of vertices in $\VV_3$ reachable from $u_1$ in $\GG\setminus \WW$ is thus $m|\cup_{i \notin B}S_i|$. Hence, the number of votes  to $1$ is $1+m|\cup_{i \notin B}S_i|$.

 We now show that $m-|B|=\ell$.
 \begin{claim}
 $m-|B|=\ell$.
 \label{clm:sc}
 \end{claim}
 \begin{proof}
For a contradiction, assume that $m-|B|\neq\ell$. We split the proof into a few cases.
\begin{description}
\item[Case 1: $m-|B|\geq \ell+1$.] In this case, the number of votes to $0$ is at least $3\ell m+1$, while the number of votes to $1$ is $1+m|\cup_{i \notin B}S_i| \leq 1 + m|\UU|=1+3\ell m$. This contradicts the hypothesis that candidate 1 uniquely wins in the restricted election.
\item[Case 2: $m-|B|\leq \ell-1$.] In this case, the number of votes to $1$ at most $1+3m(m-|B|)$ (since each $S_i$ has 3 elements), which is at most $3\ell m - 3m + 1$, while the number of votes to $0$ is at least $ 3\ell m - \ell > 3\ell m - 3m + 1$ (since $3m-1\geq m \geq \ell$). This contradicts the hypothesis that candidate 1 uniquely wins in the restricted election.\qedhere
\end{description}
 \end{proof}
 Now, for $1$ to be the unique winner in the restricted election, we have that
 \[1+m|\cup_{i \notin B}S_i| \geq 1+3\ell m - \ell +(m-|B|)=1+3\ell m,\]
 where the second equality follows from \Cref{clm:sc}. We thus have that $|\cup_{i \in [m]\setminus B}S_i|=3\ell=|\UU|$. We conclude that $\cup_{i \in [m]\setminus B}S_i=\UU$. Since $m-|B|=\ell$ by \Cref{clm:sc}, we conclude that $(\ell, S_1, \ldots, S_m)$ is a \emph{yes} instance of \XThC.
 \item[(Only If)]Let $(\ell, S_1, \ldots, S_m)$ be a \emph{yes} instance of \XThC. Thus there exists a subset $\mathsf{SC} \subseteq [m]$, such that $|\mathsf{SC}|=\ell$ and $\cup_{i \in \mathsf{SC}} S_i=\UU$. Define $A\coloneqq \{v_i \in \VV_2 \mid i \in \mathsf{SC}\}$ and $\WW\coloneqq \VV_2 \setminus A$. We now take stock of the vertices reachable from $x=u_1$ in the graph $\GG \setminus \WW$.
 \begin{itemize}
 \item Each vertex $u_i \in \VV_1$ is reachable from $u_1$ by the path $(u_1, \ldots, u_i)$.
 \item $r$ is reachable from $u_1$ by the path $(u_1,\ldots, u_{\ell(3m-1)}, r)$.
 \item Each vertex $v_i$ in $\VV_2 \setminus \WW=\VV_2\cap A$ is reachable from $u_1$ by the path $(u_1,\ldots, u_{\ell(3m-1)}, r, v_i)$. Note that no vertex in $\WW=\VV\setminus A$ exists in $\GG \setminus \WW$.
 \item We show that for every $j \in [m]$ and $k \in [3\ell]$, $w_{j, k} \in \VV_3$ is reachable from $u_1$. By the definition of $A$, for each $k \in [3\ell]$ there exists $v_i \in A$ such that $k \in S_i$. Thus, by the definition of $\EE$, $\{v_i, w_{j, k}\} \in \EE$. This implies that $w_{j, k}$ is reachable from $u_i$ by the path $(u_1,\ldots, u_{\ell(3m-1)}, r, v_i, w_{j, k})$.
 \end{itemize}
 From the above, we calculate the number of vertices $z$ reachable in $\GG \setminus \WW$ from $u_1$ with $\tau(z)=0$ to be $|\VV_1|+|\VV_2 \cap A|=|\VV_1|+|A|=\ell(3m-1)+\ell=3\ell m$. On the other hand, the number of vertices $z$ reachable in $\GG \setminus \WW$ from $u_1$ with $\tau(z)=1$ is $|\mathsf{R}|+|\VV_3|=1+3\ell m$. Thus, Candidate $1$ wins uniquely in the election induced by the restriction of the function $\tau()$ to the graph $\GG \setminus \WW$. We conclude that $(\CC,\VV,\tau,\GG,c, x)$ is a \emph{yes} instance of \BCN.
 \end{description}
 \paragraph*{Remark.} This reduction does not make use of the assumption about a \XThC instance that each element of the universe belongs to exactly two of the given subsets.
 \end{proof}

\section{Hardness of \BCN for trees}\label{sec:tree}
In this section we prove \Cref{thm:hardness-tree}.
\begin{proof}[Proof of \Cref{thm:hardness-tree}]
We will reduce \XThC (see \Cref{dxef:xthc}) to \BCN such that an arbitrary instance of \XThC reduces to an instance of \BCN where the graph is a tree.
 \subsection*{The reduction}
 Given an instance $(\ell, S_1, \ldots, S_m)$ of \XThC, we will construct an instance $(\CC,\VV,\tau,\GG,c,x)$ of \BCN, such that the graph \GG is a tree. There will be $3\ell$ candidates $1,\ldots,3\ell$, which will be identified with the elements of $\UU$. In addition, there will be two special candidates, $c$ and $d$. Thus $\CC=\UU \cup \{c, d\}$.  $c$ is the preferred candidate.

 We now proceed to the description of $\GG$. A \emph{path graph} formed of a sequence of vertices $(x_1,\ldots,x_k)$ is a graph that contains an edge between $x_i$ and $x_{i+1}$ for each $i \in [k-1]$, and contains no other edge. We will now introduce some path graphs (henceforth referred to as paths) on disjoint sets of vertices in $\GG$. For each $i \in [m]$, let $e^{(i)}_1,\ldots,e^{(i)}_{|S_i|}$ be an arbitrary ordering of the elements of $S_i$. For each $i \in [m]$ we introduce a path $P^{(i)}$ formed by a sequence of vertices $(v^{(i)}_1, \ldots, v^{(i)}_{|S_i|}, u^{(i)}_1,\ldots, u^{(i)}_{3\ell}, w^{(i)})$. For $j \in [|S_i|], k \in [3\ell]$, the vertex $v^{(i)}_{j}$ corresponds to $e^{(i)}_j \in S_i$ and the vertex $u^{(i)}_{k}$ corresponds to $k \in \UU$.

 Next, we introduce a path $P^{\#}$ formed by a sequence of two vertices $(z^{\#}_1, z^{\#}_2)$, and a path $P^*$ formed by a sequence of vertices $(z^{*}_1,\ldots,z^{*}_{m+1-\ell})$.

 For two vertices $t_1$ and $t_2$ in a path, $t_1$ will be said to be \emph{higher (lower)} than $t_2$ if $t_1$ appears before (after) $t_2$ in the sequence the path is formed of. The highest element of a path will be called the \emph{head} of the path. Thus $v^{(i)}_1$ is the head of $P^{(i)}$, $z^{\#}_1$ is the head of $P^{\#}$ and $z^*_1$ is the head of $P^*$.

 Next, we connect $z^*_{m+1-\ell}$ with the head of each path other than $P^*$, namely, vertices $v_1^{(1)},\ldots,v_1^{(m)}$ and $z_1^{\#}$. This completes the description of $\GG$. Clearly, $\GG$ is a tree.

 The voter $x$ which conducts the election is the vertex $z^*_1$.

 Next, we define $\tau$.
 \begin{enumerate}
 \item For each $i\in [m], j \in [|S_i|]$ and $k \in [3\ell]$, define $\tau(v^{(i)}_j) \coloneqq e^{(i)}_j,\tau(u^{(i)}_k) \coloneqq k$ and $\tau(w^{(i)}) \coloneqq c$.
 \item Define $\tau(z^{\#}_1) \coloneqq c$ and $\tau(z^{\#}_2) \coloneqq c$.
 \item For each $i \in [m+1-\ell]$, define $\tau(z^*_i) \coloneqq d$.
 \end{enumerate}
 This concludes the definition of the reduced instance $(\CC,\VV,\tau,\GG,c, x)$ of \BCN. It is clear that the reduction runs in polynomial time. Refer to \Cref{R3X_to_tree} for a graphical illustration of the reduction stated above.
 \begin{figure}[!htbp]
    \centering
    \includegraphics[scale=0.6]{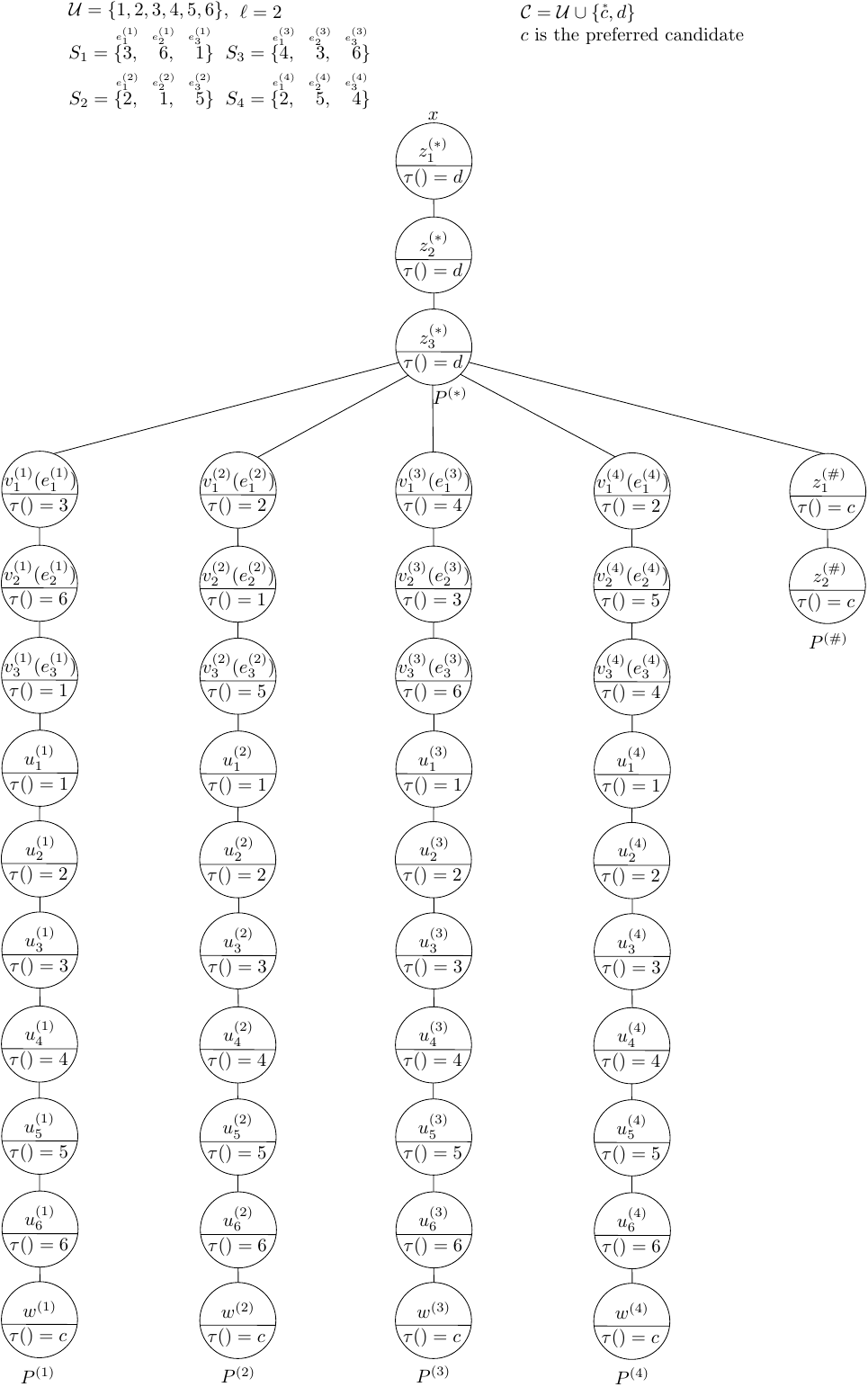}
    \caption{Reduction of \BCN from \XThC where \GG is a tree}
    \label{R3X_to_tree}
\end{figure}
 \subsection*{Correctness of the reduction}
 Let $(\ell, S_1, \ldots, S_m)$ and $(\CC,\VV,\tau,\GG,c, x)$ be as defined above. We first count the number of votes gotten by each candidate in $\GG$. Each $e \in \UU$ is contained in exactly two sets, say $S_{e_1}$ and $S_{e_2}$. So, there exists $k_{e_1} \leq |S_{e_1}|$ and $k_{e_2} \leq |S_{e_2}|$ such that $e^{(e_1)}_{k_{e_1}}=e$ and $e^{(e_2)}_{k_{e_2}}=e$. Thus $e$ gets two votes from vertex $v^{(e_1)}_{k_{e_1}}$ of $P^{(e_1)}$ and vertex $v^{(e_2)}_{k_{e_2}}$ of $P^{(e_2)}$. In addition, each $i \in [m]$ gets one vote from vertex $u^{(i)}_e$ in $P^{(i)}$. Thus, in total, $e$ gets $m+2$ votes.

 The set of voters voting for $c$ is $\{w^{(i)}\mid i \in [m]\} \cup \{z^{\#}_1, z^{\#}_2\}$. Thus, their number is $m+2$.

 The set of voters voting for $d$ is $\{z^*_i \mid i \in [m+1-\ell]\}$. Thus, their number is $m+1-\ell$.

 We now show that $(\ell, S_1, \ldots, S_m)$ is a \emph{yes} instance of \XThC \emph{if and only if}
 $(\CC,\VV,\tau,\GG,c,x)$ is a \emph{yes} instance of \BCN.
 \begin{description}
 \item[(If)] Let $(\CC,\VV,\tau,\GG,c,x)$ be a \emph{yes} instance of \BCN. Thus, there exists a set $\WW \subseteq \VV \setminus \{ z^*_1 \}$ such that candidate $c$ wins uniquely in the election restricted to the vertices reachable from $x=z^*_1$ in $\GG \setminus \WW$. We observe that certain assumptions can be made about $\WW$ without loss of generality.
 \begin{itemize}
\item $\WW$ does not contain any vertex of $P^*$, as otherwise no voter who votes for $c$ will be reachable from $x$.
\item We may assume that for each $i \in [m]$, $\WW$ either contains the head $v_1^{(i)}$ of $P^{(i)}$, or does  not contain any vertex of $P^{(i)}$. To see this, first observe that if $\WW$ contains multiple vertices from $P^{(i)}$, then we may instead include only the highest of those vertices in $\WW$; its removal will already disconnect all other vertices lower than it from $x$. Next, note that if any vertex of $w^{(i)}$ is removed, $c$ loses $1$ vote. On the other hand, if the head is removed, then the votes lost by all other candidates is maximized. Since we want to ensure the victory of $c$ in the restricted election, we may assume without loss of generality that if at all any vertex in $P^{(i)}$ is in $\WW$, then it is the head.
\item We may assume that $\WW$ does not contain any vertex in $P^{\#}$. Removing any of the vertices $z^{\#}_1$ and $z^{\#}_2$ will lead $c$ to lose votes, while no other candidate loses votes.
 \end{itemize}
 In light of the above, we assume that each element of $\WW$ is the head of a path $P^{(i)}$ for some $i \in [m]$. For the rest of this part of the proof, we identify $\WW$ with the set of indices $i \in [m]$ such that the head $v^{(i)}_1$ of $P^{(i)}$ is removed from $\GG$.
We claim that $|\WW|=\ell$.
\begin{claim}
\label{clm:sc_tree}
$|\WW|=\ell$.
\end{claim}
\begin{proof}
Towards a contradiction, assume that $|\WW|\neq \ell$. We split the proof into two cases.
\begin{description}
\item[Case 1: $|\WW| \geq \ell+1$.] In this case, the number of votes won by $c$ in the restricted election is at most $m+2-(\ell+1)=m+1-\ell$. Since the number of votes gathered by $d$ is also $m+1-\ell$, $c$ is not the unique winner, and we have a contradiction.
\item[Case 2: $|\WW| \leq \ell-1$.] In this case, $|\cup_{i \in \WW}S_i|< 3\ell=|\UU|$, and hence there is an element $e \in \UU$ that is not in $\cup_{i \in \WW}S_i$. Thus, there is no vertex $v^{(i)}_j$ such that $i \in \WW, j \in [|S_i|]$ and $v^{(i)}_j$ votes for $e$ (i.e., $e^{(i)}_j = e$). Thus, the number of votes lost of $e$ due to the deletion of the paths $P^{(i)}$ for $i \in \WW$ is exactly $|\WW|\leq \ell-1$. Thus, it gathers $m+2-|\WW|$ votes in the restricted election, which is the same as the votes gathered by $c$. Thus, $c$ is not the unique winner, which leads to the desired contradiction.\qedhere
\end{description}
\end{proof}
Thus, $|\WW|=\ell$. Furthermore, the proof of case 2 of the above claim makes it clear that for $c$ to be the unique winner in the restricted election, it is necessary that for each $e \in \UU$ there is an index $i \in \WW$ such that $e \in S_i$. Thus, $\cup_{i \in \WW}S_i=\UU$. This shows that $(\ell, S_1, \ldots, S_m)$ is a \emph{yes} instance of \XThC.
 \item[(Only If)]Let $(\ell, S_1, \ldots, S_m)$ be a \emph{yes} instance of \XThC. Thus, by definition, there exists a set $A\subset [m]$, $|A|=\ell$ such that $\cup_{i \in A} S_i=\UU$. Define $\WW\coloneqq \{v^{(i)}_1 \mid i \in A\}$ to be the set of the heads of all the paths $P^{(i)}$ such that $i \in A$.

 Deletion of $\WW$ serves to remove the paths $P^{(i)}$ such that $v^{(i)}_1 \in \WW$. All other vertices in $\GG$ are reachable from $x$ in $\GG \setminus \WW$. We now count the votes lost by each candidate due to deletion of $\WW$. Each $e \in \UU$ loses $1$ vote from $u^{(i)}_e$ in each removed path $P^{(i)}$. Now, by the definition of $A$, there exists a $S_{e_1}$ such that $i \in A$ and $e \in S_{e_1}$. Thus, there exists an index $k_1 \in [|S_{e_1}|]$ such that $e^{(e_1)}_{k_1}=e$ and the vertex $v^{(e_1)}_{k_1}$ votes for $e$. Thus, the number of votes lost by $e$ is exactly $\ell+1$. Candidate $c$ loses exactly one vote for each deleted path $P^{(i)}$, namely the one by $w^{(i)}$. Thus, it loses $\ell$ votes. Candidate $d$ does not lose any votes.

 Hence, in the restricted election, $c$ gets $m+2-\ell$ votes, $d$ gets $m+1-\ell$ votes, and each $e \in \UU$ gets exactly $m+2-(\ell+1) = m+1-\ell$ votes. Thus, $c$ is the unique winner.
 \end{description}    
\end{proof}


\section{Conclusion and Open Questions}
\label{sec:conclusion}
Our paper initiates a novel study of a natural kind of election control in online elections, and introduces problems \CN and \DCN. Our paper is a comprehensive study of the computational aspects of these two problems. We present algorithms for the problems for restricted classes of inputs (where the number of candidates is bounded for \CN, and the treewidth of the network of voters is bounded for both problems), and establish the necessity of these restrictions on inputs by complementary hardness results.

Our work leaves open several interesting questions for future research. A natural follow-up question is whether the two problems admit algorithms whose time complexities have FPT (Fixed Parameter Tractable) dependencies on the parameters $w$ (treewidth) and $m$ (number of candidates) i.e, are of the form $\OO(f(w)\cdot \mathsf{poly}(n, m))$ for \DCN and $\OO(g(w, m) \cdot \mathsf{poly}(n, m))$ for \CN, where $f$ and $g$ are arbitrary computable functions. The dependencies achieved by our algorithms (\Cref{thm:easy-low-tw-dcn} and \Cref{thm:easy-low-tw-const-cand}) are ``slicewise" and fall short of being FPT.

Another interesting direction is to consider the optimizations versions of the two problems, where the election controller wishes to control the election, constructively or destructively, by incurring minimum total cost. Can we design efficient approximation algorithms for these problems, and sidestep the bounds on treewidth and number of candidates assumed in our algorithmic results (\Cref{thm:easy-low-tw-dcn} and \Cref{thm:easy-low-tw-const-cand})?




\bibliographystyle{plain} 
\bibliography{ref}


\end{document}